\documentclass{vgtc}

\graphicspath{{figs/}} 

\usepackage{times}

\usepackage{subcaption}

\usepackage{amsthm}     
\usepackage{amssymb}
\usepackage{amsmath}
\usepackage{graphicx}
\usepackage{enumitem}
\usepackage[dvipsnames]{xcolor}

\newcommand{\change}[2]{{#2}}
\newtheorem{theorem}{Theorem}

\onlineid{1336} 
\vgtccategory{Research}

\vgtcinsertpkg

\title{Topology-Preserving Meshing of Implicit Scalar Fields\\ via Monotonicity
Constraints}

\author{
  Tanner Finken\thanks{e-mail: \{finkent, josh\}@arizona.edu}\\
  \scriptsize{University of Arizona}
  \and Jixian Li\thanks{e-mail: \{jixianli, beiwang\}@sci.utah.edu}\\
  \scriptsize{University of Utah}
  \and Bei Wang\textsuperscript{\textrm{\dag}}\\
  \scriptsize{University of Utah}
  \and Hanqi Guo\thanks{e-mail: guo.2154@osu.edu}\\
  \scriptsize{Ohio State University}
  \and Joshua A. Levine\textsuperscript{\textrm{*}}\\
  \scriptsize{University of Arizona}
}

\abstract{
Topological analysis of scalar fields yields structures such as the Morse--Smale complex (MSC) that summarize salient features across multiple scales. Existing MSC extraction algorithms typically assume an explicit representation of the input field, such as a discretely sampled mesh. However, recent advances in visualization have popularized implicit field representations, for which these assumptions no longer hold. In this work, we address the problem of extracting an MSC from an implicitly defined 2D scalar field.
We present a method for constructing a triangulated piecewise-linear (PL) mesh that \change{provably preserves}{aims to preserve} the critical points of \change{the}{an} underlying implicit scalar field. Our central insight is that \change{enforcing a monotonicity condition on all mesh edges guarantees that}{if all edges are monotonic with respect to the underlying field, then} the resulting PL approximation is topologically consistent with respect to critical points. 
\change{We show how to certify and enforce this condition under minimal assumptions on the input, }{Based on this insight, we introduce a refinement procedure that mitigates monotonicity violations.} \change{r}{Requiring} only pointwise evaluations and modest mesh refinement\change{}{, the approach produces PL meshes that are correct with regards to critical points in our experiments.} 
Finally, we demonstrate that additional targeted refinement improves the geometric fidelity of MSC separatrices.
} 


\begin{document}

\firstsection{Introduction}
\maketitle
Scalar fields are widely used to represent phenomena arising from both observational data and numerical simulations. Formally, a scalar field is a function, typically assumed to be at least continuous, that assigns a scalar value to every point in a domain. In visualization, representing such fields involves a tradeoff between accuracy and storage. Broadly, representations fall into two categories: explicit and implicit. We refer to \emph{explicit} representations as those defined on meshes or grids, where interpolation reconstructs the field from discrete samples. In contrast, \emph{implicit} representations describe the field through functional forms, either globally (e.g., a single analytic expression) or via a sparse set of local primitives such as spline patches. 

Both representations have distinct advantages and limitations. Historically, visualization has favored explicit representations because they map naturally onto standard computational pipelines. For continuous fields, however, the accuracy of an explicit representation depends on sampling density; achieving high fidelity may require very dense sampling, which increases storage costs and can limit scalability. These tradeoffs have motivated renewed interest in implicit representations, including multivariate functional approximation (MFA)~\cite{peterka2018foundations} and neural models such as implicit neural representations (INRs) for volumetric data~\cite{sitzmann2020implicit, lu_compressive_2021}. Despite these advantages, extracting global feature descriptors, particularly topological structures, from implicit representations remains challenging.

Topological data analysis provides a suite of techniques for analyzing fields by computing, visualizing, and summarizing their structure. One particularly useful construct is the Morse--Smale complex (MSC). The MSC partitions the domain into regions within which all points share the same ascent to a maximum and descent to a minimum.
Most methods for extracting MSCs assume a discrete representation of the input~\cite{bremer2004topological,de2015morse,edelsbrunner2001hierarchical,gyulassy2008practical,gyulassy2019shared,maack2023parallel,robins2011theory,shivashankar2011parallel,will2024distributed}, largely because they require explicit knowledge of each point’s neighborhood. Consequently, extracting an MSC from an implicitly defined scalar field presents two primary options.

The first option is to pursue a direct extraction, which is challenging due to the oracular nature of implicit representations. Evaluating the field at a given point requires computing it on the fly, necessitating carefully designed methods for operations such as root-finding~\cite{ma_mfa_cpe}. Nevertheless, under certain assumptions about the implicit form, this approach can be made tractable. For example, in the setting of MFAs, recent work by Ma et al. has demonstrated promising results in recovering critical points~\cite{ma_mfa_cpe} and other topological descriptors such as contours, Jacobi sets, and ridge-valley graphs~\cite{ma_mfa_topology}. In the context of INRs, Feng et al.~\cite{leila_inr_topology} propose training strategies that improve topological fidelity by coupling the model with adaptive extraction of topological features.

Alternatively, one can reconstruct an explicit mesh that approximates the implicit representation and then apply existing MSC extraction methods. This approach must address a key issue of topological consistency, as the meshing process can introduce discrepancies relative to the underlying implicit function. At the same time, an explicit mesh can be valuable beyond topological analysis: many standard visualization tasks, such as cropping, contouring, and rendering, are more naturally performed on meshed data.

In this work, we focus on extracting a topologically consistent explicit representation that \change{guarantees}{seeks to recover} all critical points of the implicit function in the resulting piecewise-linear (PL) mesh, up to a prescribed resolution. 
\change{To mitigate topological inconsistencies between the mesh and the underlying function, we enforce a \textit{monotonicity} constraint on every edge of the constructed mesh.}{Our approach is motivated by the theoretical observation that critical point consistency is achieved when all mesh edges are monotonic with respect to the underlying function. Because certifying this condition is impractical for implicit fields accessible only through pointwise evaluations, we instead adopt a sampling-based strategy to detect and resolve monotonicity violations.}
We further refine the mesh to better capture the geometry of the MSC, particularly its separatrices. 
Our contributions include:
\vspace{-0.5em}
\begin{itemize}[leftmargin=*,noitemsep]
\item  \change{}{We characterize the mesh conditions required to ensure topological consistency of critical points, showing that all mesh edges must be monotonic with respect to the underlying Morse function.}
\item \change{}{Guided by this characterization, we develop a refinement algorithm that requires only pointwise evaluations of the function and its derivative to detect and resolve monotonicity violations. The resulting meshes better satisfy the theoretical conditions and produce more accurate MSC separatrices.}
\end{itemize}

\section{Background}
\label{sec:background}

We briefly review the necessary background and establish notation. Let $f \colon \mathbb{R}^2 \rightarrow \mathbb{R}$ be a 2D scalar field, represented implicitly either by an algebraic expression or an implicit neural representation (INR). Our goal is to construct an explicit representation of $f$ in the form of a triangulated mesh. Each mesh vertex $v_i$ is associated with a position $\mathbf{x}_i \in \mathbb{R}^2$ and the corresponding function value $f(\mathbf{x}_i)$. Using piecewise-linear (PL) interpolation over the mesh, we obtain an approximation $\hat{f}$ of $f$.

Our refinement strategy is designed to ensure that $\hat{f}$ captures the gradient behavior of $f$ along mesh edges. We say that a mesh is \emph{monotonic with respect to $f$} if, for every edge connecting vertices $v_i$ and $v_j$, the restriction of $f$ to the segment $\overline{\mathbf{x}_i \mathbf{x}_j}$ is monotonic.

\subsection{Implicit Neural Representations}

We briefly review the definition of an \change{implicit neural representation (INR)}{INR}. INRs use multi-layer perceptrons to encode functions by learning a mapping from spatial coordinates $\mathbf{x}$ to scalar values $f(\mathbf{x})$. Originally popularized for representing radiance fields in computer vision (e.g., NeRFs~\cite{mildenhall2021nerf}), INRs have since found widespread applications in graphics, vision, and visualization~\cite{molaei2023implicit,xie2022neural,zhang2025neural}.

INRs have emerged as a powerful approach for scalar field visualization following several key developments. Periodic activation functions introduced by Sitzmann et al.~enable INRs to represent high-frequency continuous functions and their derivatives~\cite{sitzmann2020implicit}. In this framework, the network learns a nonlinear mapping from spatial coordinates to scalar values by fitting its parameters to sampled data, resulting in a continuous function that can be queried at arbitrary locations~\cite{dupont2022data}. INRs have also gained traction as a form of lossy compression: Lu et al.~demonstrated that volumetric data can be effectively represented by storing network weights rather than raw samples~\cite{lu_compressive_2021}. Since then, numerous extensions have adapted INRs for volumetric scalar fields across different dimensions and applications~\cite{chen2025explorable,han2022coordnet,han2025dcinr,li2024improving,tang2023ecnr,weiss2022fast,wu2023interactive,wu2024distributed}.

Despite these advantages, extracting global topological structure, such as the Morse--Smale complex, from INRs remains challenging. Because the representation is implicit and accessible only through pointwise queries, recovering large-scale features typically requires dense sampling of the domain, which can be computationally expensive. Consequently, tasks that depend on global context can diminish some of the practical efficiency benefits of INRs.

\subsection{Topology Extraction}

The Morse--Smale complex (MSC) is a fundamental structure for analyzing scalar fields. We briefly review key concepts here and refer the reader to~\cite{edelsbrunner_morse} for a detailed treatment.

The MSC segments a scalar field according to gradient flow behavior. For a 2D scalar field $f$, critical points, where $\nabla f(\mathbf{x}) = 0$, are classified as minima, saddles, or maxima based on local structure. The domain is partitioned by considering gradient flow: the ascending manifold of a maximum consists of all points whose gradient ascent trajectories terminate at that maximum, while descending manifolds are defined analogously via gradient descent toward minima. When ascending and descending manifolds intersect transversely, they partition the domain into regions that define the MSC. In 2D, separatrices---gradient paths connecting saddles to extrema---form the boundaries of these regions, yielding a decomposition aligned with the global flow structure.

Bremer et al.~extend this construction to PL functions~\cite{bremer2004topological}, where MSC extraction reduces to identifying critical points and tracing separatrices. However, in the PL setting, gradients are piecewise constant, and separatrix tracing introduces additional subtleties~\cite{edelsbrunner2001hierarchical}. In this work, we adopt a strategy that traces separatrices along edges or through triangle interiors as needed, following the steepest gradient.

Although our primary focus is on critical points and the MSC, our treatment of monotonic edges also relies on reasoning about isocontours. An \emph{isocontour} at value $\rho$ is the set of points $P$ such that $\forall \mathbf{x} \in P,\ f(\mathbf{x}) = \rho$. As $\rho$ varies, connected components of isocontours form, merge, split, and disappear, with these topological changes occurring precisely at critical points. Importantly, any monotonic edge intersects a given isocontour at most once.
\section{Naive Approach}

A straightforward approach to obtaining an explicit representation is to densely sample the continuous implicit function. As the sampling density increases, the PL  approximation $\hat{f}$ converges geometrically to $f$. However, this geometric convergence does not guarantee topological consistency.

To understand why, consider the stability result of Cohen-Steiner et al.~\cite{stability_diagrams}, which bounds topological differences in terms of function distance:
\[
d_B(D(f), D(\hat{f})) \leq \lVert f - \hat{f} \rVert_\infty.
\]
Here, $D(f)$ denotes the persistence diagram of $f$, and $d_B$ is the \emph{bottleneck distance}. This result implies that as two functions become closer in the $\ell_\infty$ sense, their topological summaries also converge.

However, the bottleneck distance permits low-persistence features to be ignored, effectively matching them to the diagonal of the persistence diagram. Consequently, while the overall topology of $f$ and $\hat{f}$ may be close, they need not share the same number of critical points. In particular, spurious critical points may appear in $\hat{f}$, provided their persistence is smaller than the $\ell_\infty$ bound between the two functions. 

To illustrate this issue, we consider the function
\begin{equation} \label{eq:easy_sinc}
    f(x,y) = x \cdot \sin\!\left(x^2 - \frac{80}{x}\right) \cdot e^{-y^2},
\end{equation}
defined over the domain $[-5,-2]\times[-1.5,1.5]$.  
~\Cref{fig:sinc_poisson} demonstrates how naive sampling of \cref{eq:easy_sinc} can introduce spurious topology compared to the ground truth shown in \cref{fig:sinc_ground_truth}. In particular, these spurious critical points arise from misalignment between mesh elements and the underlying ridges, creating small pockets in which the PL approximation exhibits artificial fluctuations (\cref{fig:sinc_zoom}). Increasing the sampling density reduces the magnitude of these artifacts but does not eliminate them; in practice, it may even introduce additional ones. Achieving the correct topology instead requires careful placement of vertices and mesh elements, as illustrated in \cref{fig:sinc_ours}.

\begin{figure}
\centering
\begin{subfigure}[t]{0.2\linewidth}
    \includegraphics[trim=46 30 56 34,clip,height=0.9\textwidth]{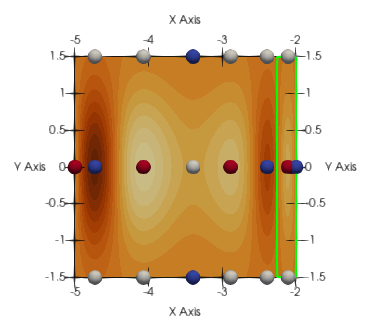}
    \vspace{-0.5em}
    \caption{\label{fig:sinc_ground_truth}}
\end{subfigure}
\hfill
\begin{subfigure}[t]{0.2\linewidth}
    \includegraphics[trim=50 31 53 32,clip,height=0.9\textwidth]{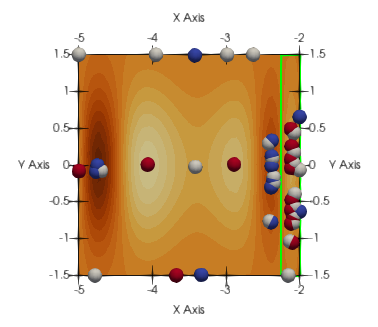}
    \vspace{-0.5em}
    \caption{\label{fig:sinc_poisson}}
\end{subfigure}
\hfill
\begin{subfigure}[t]{0.2\linewidth}
    \includegraphics[trim=46 31 55 34,clip,height=0.9\textwidth]{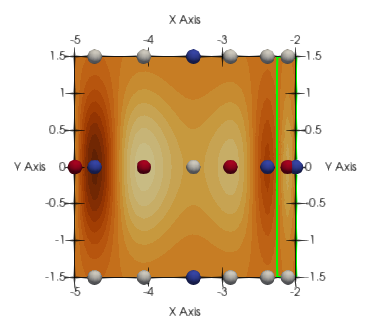}
    \vspace{-0.5em}
    \caption{\label{fig:sinc_ours}}
\end{subfigure}
\hfill
\begin{subfigure}[t]{0.09\linewidth}
    \includegraphics[width=\textwidth]{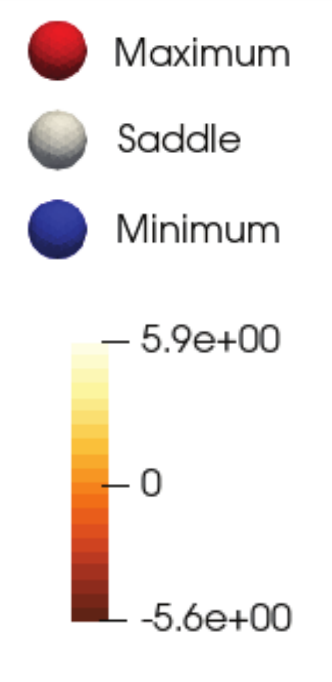}
\end{subfigure}
\hfill
\begin{subfigure}[t]{0.27\linewidth}
    \includegraphics[trim=0 0 0 60,clip,width=\textwidth]{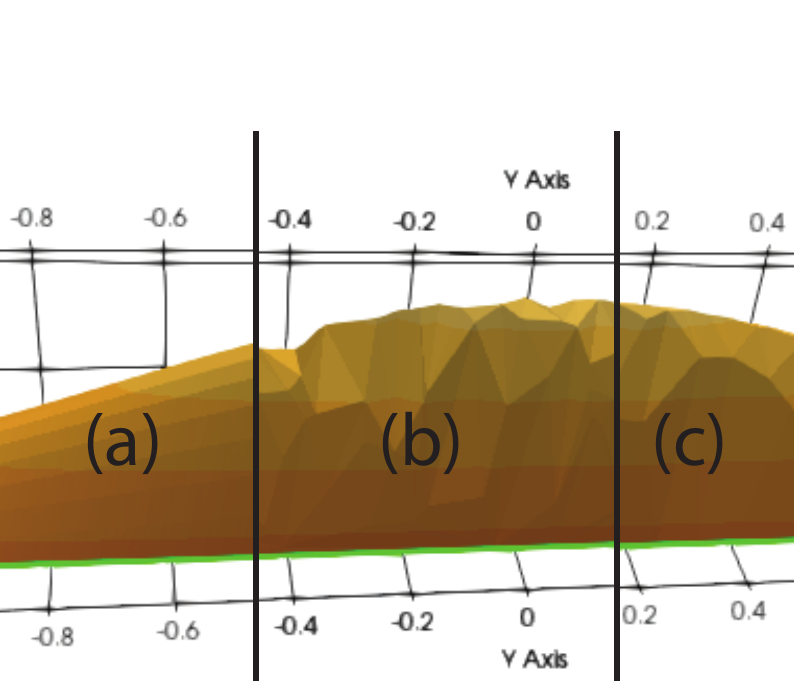} 
    \vspace{-1.5em}
    \caption{\label{fig:sinc_zoom}}
\end{subfigure}
    \vspace{-0.5em}
\caption{Meshing the implicit function defined in \cref{eq:easy_sinc}. (a) Ground truth. (b) Dense Poisson disk sampling without refinement. (c) Our method applied to Poisson disk sampling ($R=0.05$, $w=0.01$). (d) Close-up of the highlighted region, with $f(\mathbf{x})$ shown as elevation, illustrating mesh alignment artifacts in (b) and their mitigation in (c).} 
\label{fig:sinc_dataset}
\vspace{-3ex}
\end{figure}

\subsection{Monotonic Guarantees} 

We first establish a key topological result showing that PL functions defined on meshes with monotonic edges have critical points that are consistent with those of the underlying function, up to the sampling resolution. In particular, topological inconsistencies arise only when triangles are sufficiently large to contain more than one critical point.

\begin{theorem}
\label{theorem:monotonic}
Let $\hat{f}$ be a PL function defined on a mesh that is monotonic with respect to a Morse function $f$. Then:
\begin{enumerate}[leftmargin=*,noitemsep]
    \item Every critical point of $\hat{f}$ coincides with a critical point of $f$, and
    \item Any critical point of $f$ either coincides with a critical point of $\hat{f}$ or shares a triangle with at least one other critical point of $f$.
\end{enumerate}
\end{theorem}

\begin{proof} 
Our proof proceeds by case analysis over critical point types. We first consider critical points of $\hat{f}$ (Part~1), followed by critical points of $f$ (Part~2). We use standard terminology for the neighborhood of a vertex: the \emph{star} of a vertex consists of the vertex together with all incident edges and triangles, while the \emph{link} consists of the vertices and edges adjacent to the vertex but not including it, forming the one-ring around it.

\begin{figure}
\centering
\begin{subfigure}[t]{0.23\linewidth}
    \includegraphics[trim=30 40 50 40,clip,width=\textwidth]{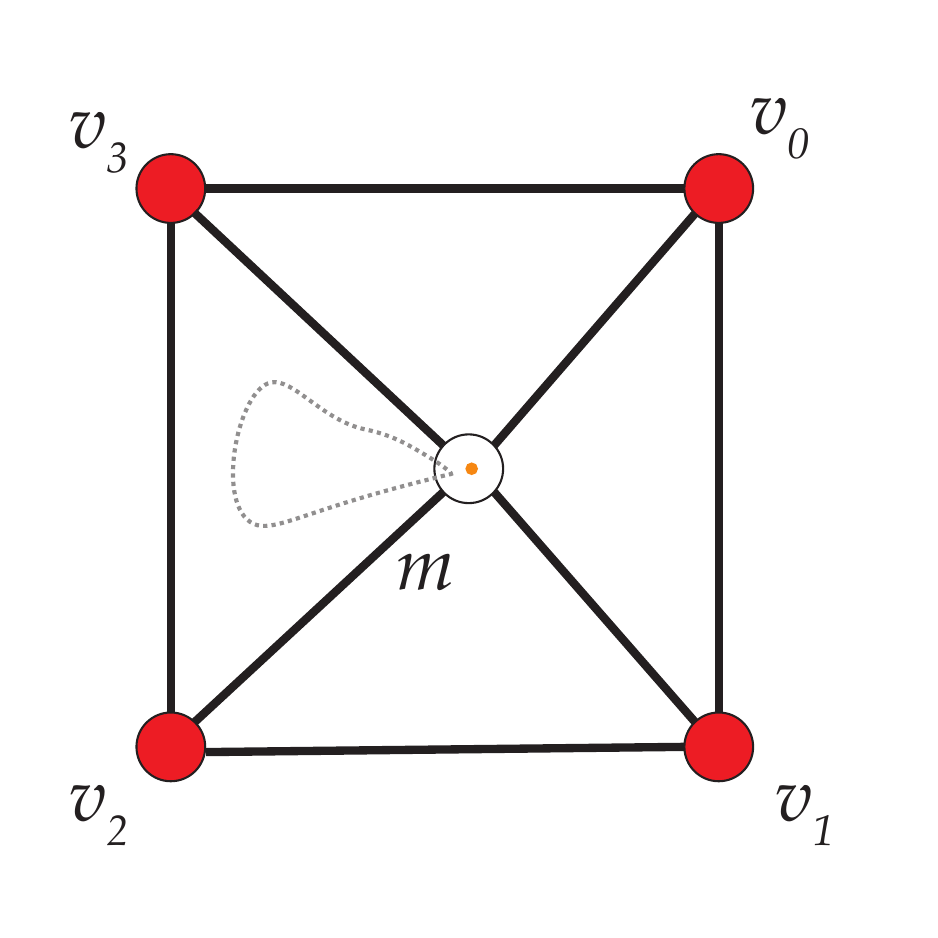}
    \vspace{-1.5em}
    \caption{\label{fig:vertex_min_example}}
\end{subfigure}
\hfill
\begin{subfigure}[t]{0.23\linewidth}
    \includegraphics[trim=30 40 50 40,clip,width=\textwidth]{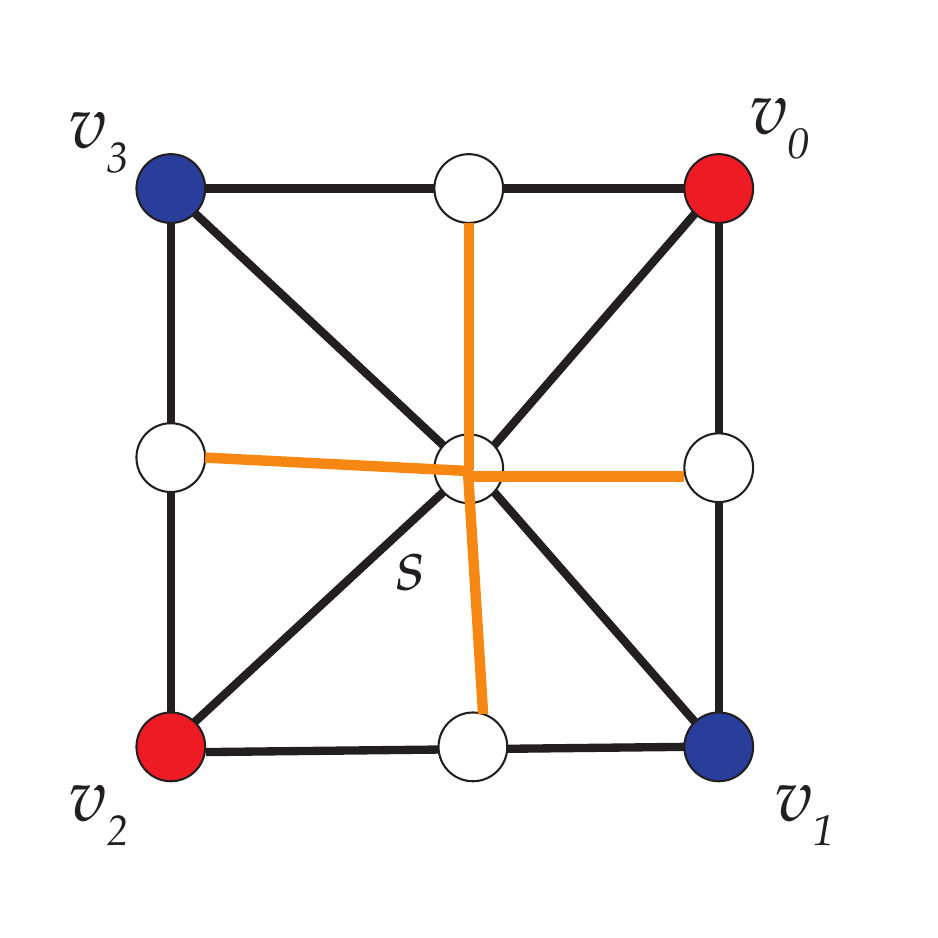}
    \vspace{-1.5em}
    \caption{\label{fig:vertex_saddle_example}}
\end{subfigure}
\hfill
\begin{subfigure}[t]{0.23\linewidth}
    \includegraphics[trim=20 20 20 20,clip,width=\textwidth]{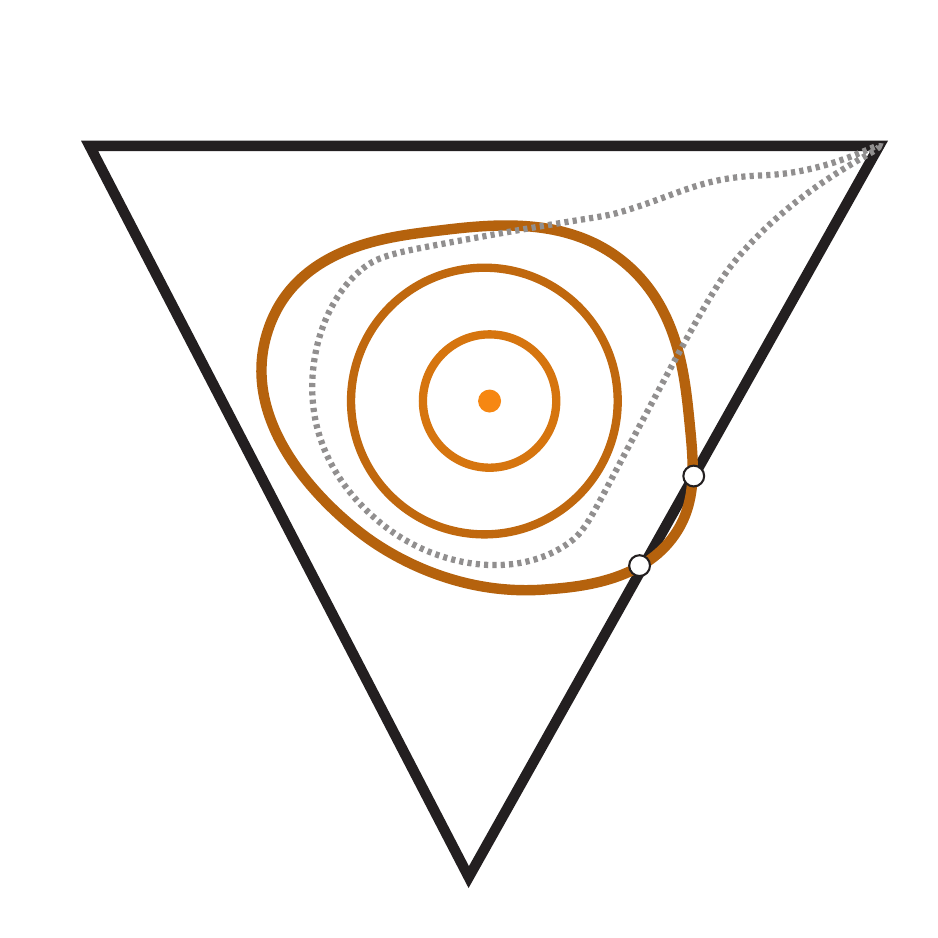}
    \vspace{-1.5em}
    \caption{\label{fig:extrema_contour}}
\end{subfigure}
\hfill
\begin{subfigure}[t]{0.23\linewidth}
    \includegraphics[trim=20 20 20 20,clip,width=\textwidth]{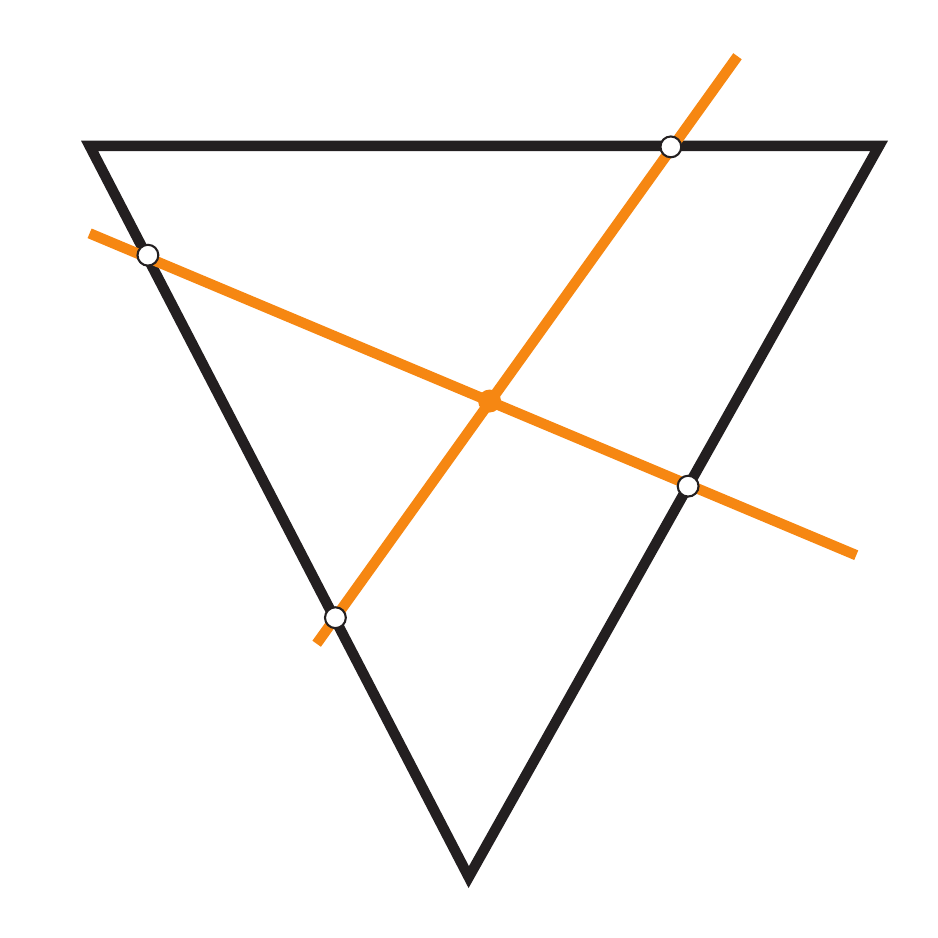}
    \vspace{-1.5em}
    \caption{\label{fig:saddle_contour}}
\end{subfigure}
\vspace{-2mm}
\caption{Cases illustrating ~\cref{theorem:monotonic}. (a) An extremum in $\hat{f}$: the corresponding isocontour in $f$ cannot follow the dotted configuration. (b) A saddle in $\hat{f}$: the isocontour in $f$ must intersect at the vertex. (c) Near an isolated extremum of $f$, isocontours must expand to cross an edge twice; the dotted configuration is not feasible. (d) For an isolated saddle in $f$, the isocontour must intersect an edge twice.}
\label{fig:monotonic_examples}
\vspace{-4mm}
\end{figure}

For Part~1, we consider critical points of $\hat{f}$, which, under PL interpolation, can only occur at vertices since $\nabla \hat{f}$ is nonzero elsewhere. We distinguish between extrema and saddles.

For an extremum, consider a minimum at $\mathbf{x}_m$ (the argument for maxima is symmetric). Let $\rho_m = f(\mathbf{x}_m)$ and consider the isocontour component $\sigma = f^{-1}(\rho_m)$ passing through $\mathbf{x}_m$. By monotonicity, all adjacent vertices have values strictly above $\rho_m$, and $f^{-1}(\rho_m)$ cannot intersect any edge in the star or link of $\mathbf{x}_m$. Thus, $\sigma$ must be contained entirely within a single triangle incident to $\mathbf{x}_m$. Since $f$ is Morse, $\sigma$ is continuous and differentiable. However, remaining within a single triangle while passing through $\mathbf{x}_m$ forces $\sigma$ to form a sharp corner at $\mathbf{x}_m$. The adjacent edges (\cref{fig:vertex_min_example}) constrain the tangent directions of $\sigma$, and unless they are colinear, this contradicts differentiability.

For a saddle at $\mathbf{x}_s$, let $\rho_s = f(\mathbf{x}_s)$ and consider the isocontours $f^{-1}(\rho_s)$. By monotonicity, these contours cannot intersect edges in the star of $\mathbf{x}_s$, but adjacent vertices may lie above or below $\rho_s$. Since $\mathbf{x}_s$ is a PL saddle, the neighboring vertices partition into two groups above $\rho_s$ and two below (\cref{fig:vertex_saddle_example}). Each edge connecting vertices from different groups contains exactly one point at which $f=\rho_s$, by monotonicity. Consequently, the contours must lie within the four incident triangles and form four branches. Because $f$ is Morse, there are exactly four such branches, and they must meet at $\mathbf{x}_s$; otherwise, they would have to cross an edge in the star, violating monotonicity.

For Part 2, it suffices to show that an isolated critical point cannot exist within a triangle-that is, no additional critical points are contained in its interior. \change{First}{Overall}, we consider an isolated critical point $\mathbf{x}$ of $f$ contained strictly within a triangle. For a minimum (the maximum case is symmetric), isocontours at values slightly above $f(\mathbf{x})$ form closed loops within the triangle, as no other critical points are present. As these contours expand, they must eventually intersect an edge. Continuing to increase the value forces multiple intersections with the same edge, violating monotonicity (\cref{fig:extrema_contour}). 

For a saddle of $f$, the four contour branches expand outward from $\mathbf{x}$. In the absence of additional critical points within the triangle, these branches must eventually intersect the triangle boundary. Since a triangle has only three edges, at least two branches must intersect the same edge, again violating monotonicity (\cref{fig:saddle_contour}).
\end{proof}

\vspace{-0.5em}
\section{Method}

We have just shown that with monotonic edges in a PL mesh, \change{it can be certified that all the}{no} critical points in the PL mesh are \change{not}{}spurious.  
\change{To achieve a mesh with monotonic edges}{To remove non-monotonic edges}, we employ a Delaunay refinement method.  
These methods construct a mesh by initializing a sample of points, and then iteratively inserting points until conditions are satisfied.  
At each insertion, the mesh is updated so that it maintains the Delaunay criteria.  
Delaunay refinement has a rich history that extends back to work from Paul Chew~\cite{chew1989guaranteed}.  
Cheng et al.'s book~\cite{cheng2013delaunay} provides a much more complete reference to many of the developments since.  
\change{In our work, w}{W}e utilize CGAL's implementation of the Delaunay triangulation~\cite{cgal:r-ctm2-26a}. 

Enforcing monotonicity in practice may require precise placement of points.
Instead, our refinement approach solves for approximate positions, with two separate applications of Newton's method.
We use four stages of refinement, outlined in \cref{fig:pipeline}.


\begin{figure} 
    \centering
    \includegraphics[trim=0 0 0 0,clip,width=0.97\linewidth]{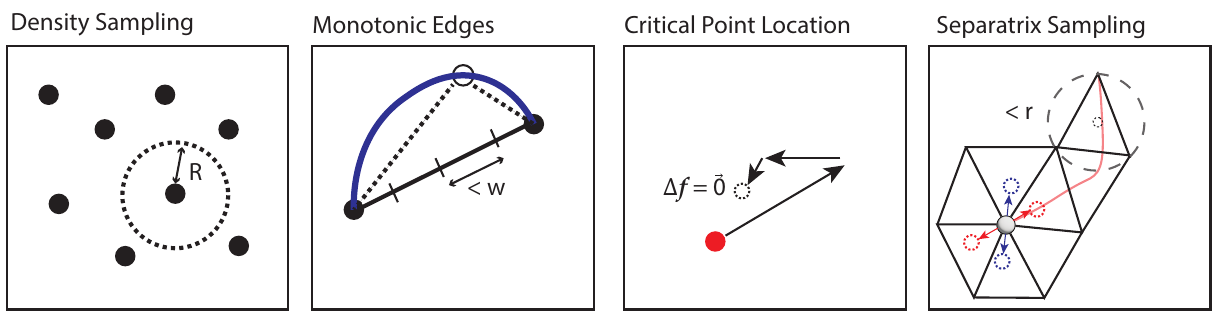}
    \vspace{-2mm}
    \caption{Pipeline overview. Our approach proceeds in four stages: (1) sample points to achieve the desired density, ensuring all triangles are sufficiently small; (2) insert additional vertices to resolve non-monotonic edges; (3) update the positions of PL critical points; and (4) refine the mesh by inserting vertices so that PL separatrices better align with their implicit counterparts.}
    \label{fig:pipeline}
    \vspace{-4mm}
\end{figure}



\textbf{Density Sampling.\hspace{0.5em}}
\change{}{Density sampling aims to have critical points isolated using appropriately sized triangles.} To \change{sample at}{identify} the desired density, we introduce a parameter $R$ to control how densely sampled the domain should be. Our algorithm works by either using a regular grid sampling or a Poisson disk sampling~\cite{bridson2007fast} to ensure 
that triangles do not contain pairs of critical points within a distance of \change{$R$}{$2R$}. 
For Poisson disk sampling, there is slightly more computational cost, but it can produce more variety in edge directions, leading to more directions to follow for separatrix tracing~\cite{le2024revisiting}. 
Note that any sampling strategy can work as long as it ensures \change{an isolation condition}{small enough triangles (i.e., circumradius $<R$)}. 
\textbf{Monotonic Edges.\hspace{0.5em}}
\change{}{Then, by~\cref{theorem:monotonic}, we know that having only monotonic edges in the PL mesh both removes spurious critical points and any missing critical points must be at least paired within a triangle. Therefore, we want to adjust the mesh to form monotonic edges so spurious critical points are removed and isolated critical points are identified.}
%
%
%
\change{W}{To do this, w}e look at each edge in the current mesh, and based on a length scale $w$, we sample $\nabla f$ to detect non-monoticity.
For an edge $e$, we sample uniformly $\max(2,\lceil \frac{||e||}{w}\rceil+1)$ times to at least sample the endpoints of the edge and guarantee a sample rate $<w$. 
The edge is \change{determined}{detected} to be non-monotonic when any two sample gradients, projected in the direction of the edge, disagree in orientation.
We then use a 1D version of Newton's method to approximate the location where orientation of $\nabla f$ switches, and insert a new vertex to split the edge into monotonic sections.
Insertion will inevitably create new edges, and so we repeat this process on each newly created edge.


To identify points to insert, we restrict $f$ to a 1D domain parameterized by $t \in [0,1]$ for which $x(t), y(t)$ interpolate between the two endpoints of the edge.  
Let $g(t) = f(x(t), y(t))$ be this restriction.
Applying the chain rule, the \change{first and second}{}derivatives of $g$ are:
$g'(t) = f_x \frac{dx}{dt} + f_y\frac{dy}{dt},
g''(t) = f_{xx}\left(\frac{dx}{dt}\right)^2 + 2f_{xy}\frac{dx}{dt}\frac{dy}{dt} + f_{yy}\left(\frac{dy}{dt}\right)^2. $
\vspace{0.25em}

Newton's method is then applied in this 1D setting to refine the parameter $t$ \change{until we}{aiming to} identify the point for which $g'(t) = 0$\change{:}{.}
The insertion of the points where $g'(t)=0$ will typically correspond to changes in gradient behavior. 
\change{We skip insertions when we do not find both a positive and negative value while sampling }{When non-monotonicity is not detected, the gradient change occurs at a length smaller than $w$.}

\textbf{Critical Point Improvement.\hspace{0.5em}}
\change{While~\cref{theorem:monotonic} ensures a correspondence between critical points, d}{D}ue to edges being only approximately monotonic, the exact locations of critical points may slightly differ between $f$ and $\hat{f}$.
We again apply Newton's method, this time to locate critical points by solving for $\nabla f(\mathbf{x}) = 0$, yielding a multivariate update
\[
\mathbf{x}_{n+1} = \mathbf{x}_n - H(\mathbf{x}_n)^{-1} \nabla f(\mathbf{x}_n),
\]
where $H(\mathbf{x}_n)$ is the Hessian matrix. This technique is applied to all critical points of $\hat{f}$ to improve their positions.

\textbf{Separatrix Refinement.\hspace{0.5em}}
When computing the separatrix lines of $\hat{f}$, we follow the edge or triangle of steepest ascent/descent through the mesh~\cite{bremer2004topological}. 
However, these directions only approximate the integral lines of $f$ emanating from the saddles.  
We can improve the separatrices of $\hat{f}$ by tracing and adding more samples around these paths initialized around saddles.
Each saddle has two ascending directions and two descending directions corresponding to the eigenvectors of the Hessian matrix at that point.
We begin by inserting points in the eigenvector directions as shown in the 4th step of ~\cref{fig:pipeline} and then tracing steepest paths.
We use a parameter $r$ to adaptively adjust mesh density in regions near PL separatrices. 
At each step on the path, all adjacent triangles are refined to have a circumradius $<r$ by iteratively inserting the circumcenters. Additionally, when following a steepest gradient through a triangle, the point on the opposite edge is inserted and followed.
The additional points allow $\hat{f}$'s steepest directions to more closely align with $f$'s integral lines in these regions. 
\section{Experiments}

The algorithm was implemented using Python and results visualized using ParaView. The code \change{will be made}{is} publicly available on GitHub \change{upon acceptance}{at \href{https://github.com/finkentA/INR_MSC_Code}{https://github.com/finkentA/INR\_MSC\_Code}}.

\subsection{Synthetic Function Test} 
To show our method's ability to capture critical points and improve separatrix geometry at higher resolutions (denser sampling), we consider a synthetic function formed by summing together a Griewank function and Gaussian over the range $[-6,6]\times [-6,6]$:
\begin{equation} \label{eq:grewank}
    f(x,y) = \frac{x^2 + y^2}{4000} - \cos(x) \cdot \cos\left(\frac{y}{\sqrt{2}}\right) + 1 + \frac{5}{2\pi} \cdot e^{-\frac{(x-1.5)^2+y^2}{0.02}}
\end{equation} 
Since the Gaussian's variance is relatively low, its critical point pair around $(1.5,0)$ may not be detected at lower resolutions. We ran our method at 3 different resolutions by varying parameters as described in~\cref{fig:gwank_dataset}. At the lowest resolution (\cref{fig:gwank_lowres}) the Gaussian critical point is not detected, but it is captured at higher resolutions.  Additionally, we utilized the separatrix refinement sampling to improve the separatrices in~\cref{fig:gwank_medres} to be similar to the higher resolution ones in~\cref{fig:gwank_highres} showing a comparable topology, albeit with 4567 fewer vertices (-42\%) in the resulting mesh. 

\begin{figure} 
\centering
\begin{subfigure}[t]{0.30\linewidth}
    \center
    \includegraphics[trim=0 60 110 0,clip,height=60pt]{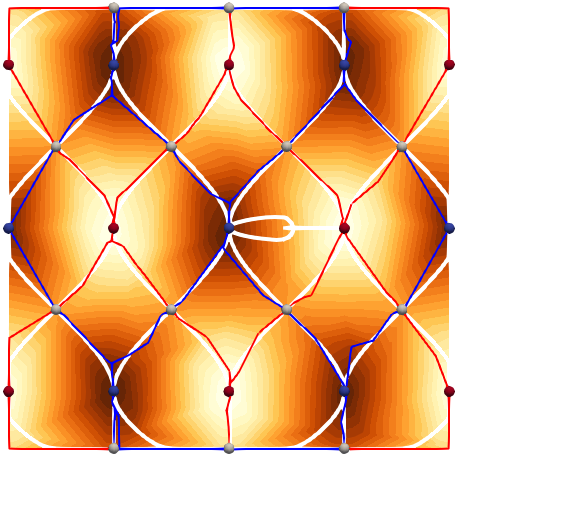} 
    \caption{ Low Resolution }
    \label{fig:gwank_lowres}
\end{subfigure}
\hfill
\begin{subfigure}[t]{0.30\linewidth}
\center
    \includegraphics[trim=0 60 110 0,clip,height=60pt]{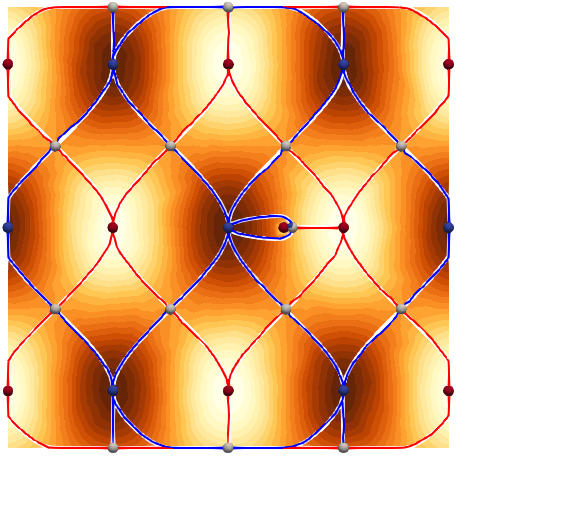} 
    \caption{ Medium, refined }
    \label{fig:gwank_medres}
\end{subfigure}
\hfill
\begin{subfigure}[t]{0.30\linewidth}
\center
    \includegraphics[trim=0 60 111.5 0,clip,height=60pt]{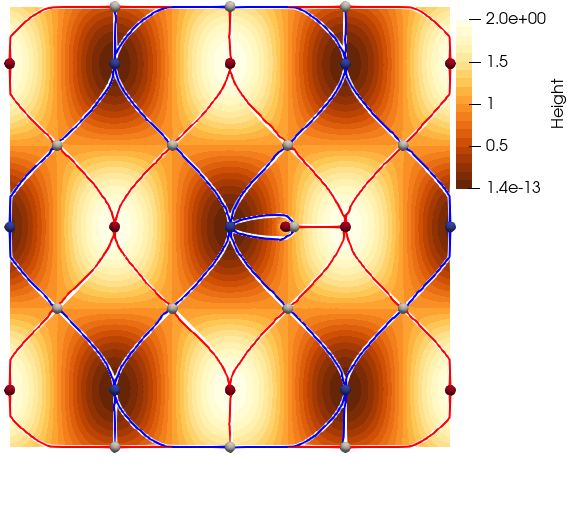} 
    \caption{ High Resolution }
    \label{fig:gwank_highres}
\end{subfigure}
\begin{subfigure}[t]{0.07\linewidth}
    \includegraphics[trim=220 60 227 0,clip,height=60pt]{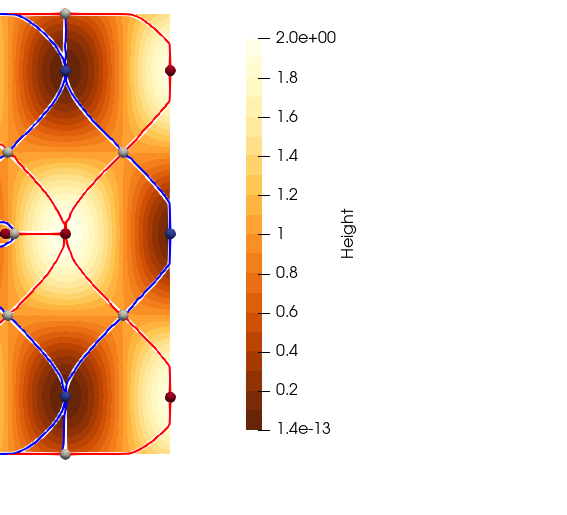} 
\end{subfigure}
\vspace{-2mm}
\caption{Griewank and Gaussian dataset (\cref{eq:grewank}) sampled via Poisson disk sampling ($w=0.1$) at varying resolutions: (a) $R=2.0$ (502 points), (b) $R=1.0$ with separatrix refinement $r=0.15$ (6,256 points), and (c) $R=0.1$ (10,823 points). Extracted PL separatrices are shown in red (ascending) and blue (descending), with the \change{}{assumed} ground truth overlaid in white.} 
\label{fig:gwank_dataset}
\end{figure}



\subsection{INR Test Case}

To demonstrate our goal of extracting topology from implicit neural representations (INRs), we apply our method to a previously trained INR on terrain data from~\cite{leila_inr_topology}.
For comparison, we first consider a regular grid derived from the INR (of size $500 \times 500$ with diagonals from top left to bottom right). We then use a small amount of persistence-based simplification ($0.025\%$) to eliminate any spurious critical points caused by meshing and not having monotonic edges. The resulting critical points are interpreted as the salient features that characterize the underlying topology of the INR. 
Notably, the INR may also encode additional topological structures that are either not captured by the uniform discretization or are removed during simplification due to very low persistence.
\Cref{fig:implicit_ours} presents the results of applying our method to this INR and simplified by the same amount, $0.025\%$, with the reference critical points marked as background squares. Our approach successfully matches all reference critical points with a much smaller mesh (6,359 vertices). 

\begin{figure}
\vspace{-2mm}
\centering
\begin{subfigure}[t]{0.8\linewidth}
    \includegraphics[trim=0 0 120 0,clip,height=75pt]{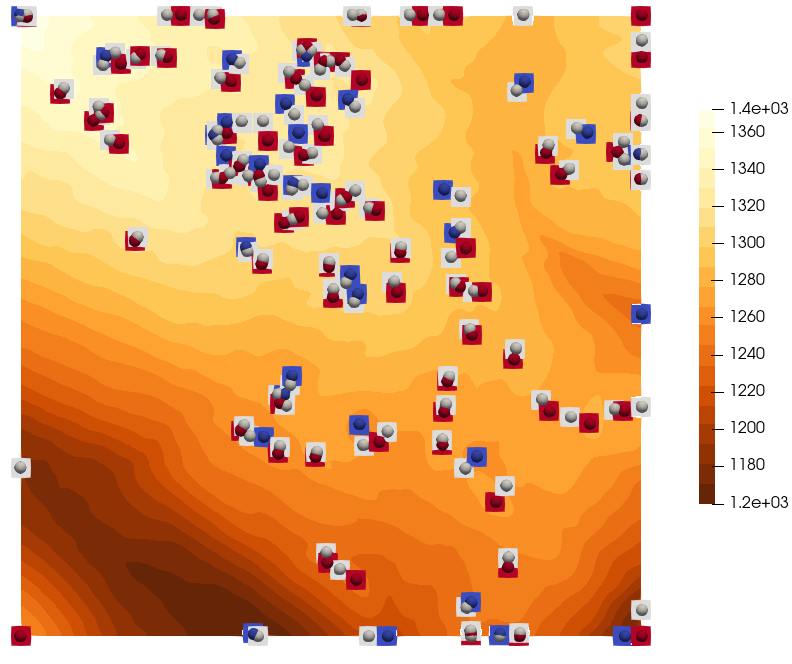} 
\end{subfigure}
\hspace{-13em}
\begin{subfigure}[t]{0.08\linewidth}
    \includegraphics[height=75pt]{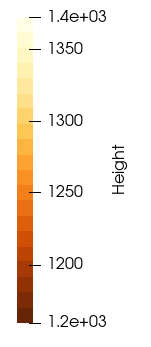} 
\end{subfigure}
\vspace{-2mm}
\caption{Implicit Terrain INR dataset~\cite{leila_inr_topology} reconstructed using our method. Parameters: Poisson disk sampling with $R=0.05$ and $w=0.001$. Critical points from our method are shown as spheres, while those from a uniform mesh ($500 \times 500$ grid with consistent diagonals) are shown as background squares. Both are simplified using a $0.025\%$ persistence threshold.}
\label{fig:implicit_ours}
\vspace{-4mm}
\end{figure}


\section{Discussion and Future Work}
\label{sec:discussion}

Overall, we present a method \change{for constructing}{aiming to construct} a topologically consistent PL mesh from a 2D implicit function by \change{enforcing monotonicity}{detecting non-monotonicity} along mesh edges. This approach is particularly well suited to INRs, enabling reliable topology extraction and supporting the analysis of topology preservation during training.

Our formulation does not address boundary critical points, as these do not correspond to true critical points of the underlying continuous function (i.e., where the gradient vanishes). A further limitation is the restriction to 2D; extending the framework to 3D is an important direction for future work. Parameter selection can also be challenging when the function is not explicitly known. However, the observed stability of topological features across resolutions suggests an effective iterative strategy, where parameters are progressively refined to uncover previously unresolved features.

\change{}{The primary computational cost of our method stems from repeated evaluations of the INR and its derivatives. Unlike the naive approach, which evaluates only the function value at each vertex, our method repeatedly evaluates the function and its derivatives along mesh edges to detect and resolve monotonicity violations. Consequently, although the resulting mesh is typically sparser, it requires substantially more INR evaluations to determine point placements that better satisfy the desired topological guarantees. While batched evaluations and query caching can partially mitigate this overhead, inherently iterative procedures—such as Newton's method and the repeated verification of newly formed edges during refinement—remain computationally expensive.}
Additional overhead arises from constructing the Delaunay triangulation, but this representation enables a wide range of downstream mesh-based analyses. Future work includes improving the efficiency of separatrix refinement, extending the approach beyond scalar fields, and incorporating additional information from the continuous function to recover richer topological structure.

\acknowledgments{This research was supported by the U.S. Department of Energy under grants DE-SC0023319 (University of Arizona), DE-SC0023157 (University of Utah), and DE-SC0022753 (The Ohio State University). 
The authors acknowledge Yansong Yu for their early exploratory work that helped shape this study.
}
\vspace{-0.5em}


\bibliographystyle{plain} 
\bibliography{references} 

\end{document}